\documentclass[a4paper,fleqn]{cas-sc}

\usepackage[numbers]{natbib}
\usepackage{amsmath}
\usepackage{amssymb}
\usepackage{amsthm}
\usepackage{array}
\usepackage{ascmac}
\usepackage{graphicx}
\usepackage[algoruled,linesnumbered,algo2e,vlined]{algorithm2e}
\usepackage{algorithmicx}
\usepackage{algpseudocode}
\usepackage{url}
\usepackage{multicol}
\usepackage{cases}
\usepackage[normalem]{ulem}
\usepackage{boites}
\usepackage{here}
\usepackage[mathlines]{lineno}
\usepackage[final,mode=multiuser]{fixme}

\ExplSyntaxOn
\cs_if_exist:NF \vbox_unpack_clear:N
  { \cs_new_protected:Npn \vbox_unpack_clear:N #1 { \box_use_drop:N #1 } }
\ExplSyntaxOff

\fxsetup{theme=color}
\definecolor{fxtarget}{rgb}{0.0000,0.0000,0.4823}
\fxuseenvlayout{colorsig}
\fxusetargetlayout{color}
\FXRegisterAuthor{hu}{ahu}{HU}
\FXRegisterAuthor{hs}{ahs}{HS}
\FXRegisterAuthor{yn}{ayn}{YN}
\FXRegisterAuthor{si}{asi}{SI}
\FXRegisterAuthor{hb}{ahb}{HB}
\FXRegisterAuthor{dk}{adk}{DK}
\fxsetface{env}{}

\newtheorem{theorem}{Theorem}
\newtheorem{lemma}[theorem]{Lemma}
\newtheorem{corollary}[theorem]{Corollary}
\newtheorem{observation}[theorem]{Observation}
\newtheorem{fact}[theorem]{Fact}
\newtheorem{problem}[theorem]{Problem}

\newcommand{\LPNF}{\mathsf{LPnF}}

\newcommand{\LNF}{\mathsf{LNF}}

\newcommand{\Substr}{\mathsf{Substr}}
\newcommand{\Suffix}{\mathsf{Suffix}}

\newcommand{\occ}{\mathsf{occ}}
\newcommand{\cocc}{\mathsf{cocc}}
\newcommand{\ncocc}{\mathsf{ncocc}}

\newcommand{\C}{\mathsf{C}}
\newcommand{\N}{\mathsf{N}}
\newcommand{\U}{\mathsf{U}}
\newcommand{\CC}{\mathcal{C}}
\newcommand{\NN}{\mathcal{N}}

\newcommand{\DD}{\mathcal{D}}

\newcommand{\Runs}{\mathsf{Runs}}
\newcommand{\dup}{\mathsf{dup}}
\newcommand{\Truns}{T_{\mathsf{runs}}}
\newcommand{\TLPnF}{T_{\mathsf{LPnF}}}
\newcommand{\TST}{T_{\mathsf{ST}}}

\newcommand{\STree}{\mathsf{STree}}
\newcommand{\str}{\mathsf{str}}

\begin{document}
\let\WriteBookmarks\relax
\def\floatpagepagefraction{1}
\def\textpagefraction{.001}

\shorttitle{Counting Distinct (Non-)Crossing Substrings in Optimal Time}
\shortauthors{Umezaki et~al.}

\title[mode=title]{Counting Distinct (Non-)Crossing Substrings in Optimal Time}

\author[kyushu-ist]{Haruki Umezaki}
\ead{umezaki.haruki.314@s.kyushu-u.ac.jp}

\author[jgmi]{Hiroki Shibata}[orcid=0009-0006-6502-7476]
\ead{shibata.hiroki.753@s.kyushu-u.ac.jp}

\author[yamanashi]{Dominik K\"oppl}[orcid=0000-0002-8721-4444]
\ead{dkppl@yamanashi.ac.jp}

\author[kyushu-inf]{Yuto Nakashima}[orcid=0000-0001-6269-9353]
\ead{nakashima.yuto.003@m.kyushu-u.ac.jp}

\author[kyushu-inf]{Shunsuke Inenaga}[orcid=0000-0002-1833-010X]
\ead{inenaga.shunsuke.380@m.kyushu-u.ac.jp}

\author[science-tokyo]{Hideo Bannai}[orcid=0000-0002-6856-5185]
\ead{hdbn.dsc@tmd.ac.jp}

\affiliation[kyushu-ist]{organization={Department of Information Science and Technology, Kyushu University},
            city={Fukuoka},
            country={Japan}}

\affiliation[jgmi]{organization={Joint Graduate School of Mathematics for Innovation, Kyushu University},
            city={Fukuoka},
            country={Japan}}

\affiliation[yamanashi]{organization={Department of Computer Science and Engineering, University of Yamanashi},
            city={Kofu},
            country={Japan}}

\affiliation[kyushu-inf]{organization={Department of Informatics, Kyushu University},
            city={Fukuoka},
            country={Japan}}

\affiliation[science-tokyo]{organization={M\&D Data Science Center, Institute of Integrated Research, Institute of Science Tokyo},
            city={Tokyo},
            country={Japan}}

\begin{abstract}
Let $w$ be a string of length $n$.
    The problem of counting factors crossing a position ---
    Problem 64 from the
    textbook ``125 Problems in Text Algorithms'' [Crochemore, Lecroq, and Rytter, 2021] ---
    asks to count the number $\CC(w,k)$ (resp. $\NN(w,k)$)
    of distinct substrings in $w$
    that have occurrences containing (resp. not containing) a position $k$ in $w$.
    The solutions provided in their textbook compute
    $\CC(w,k)$ and $\NN(w,k)$ in $O(n)$ time \emph{for a single position} $k$ in $w$, and thus a direct application would require $O(n^2)$ time for \emph{all positions} $k = 1, \ldots, n$ in $w$. Their solution is designed for constant-size alphabets.
    In this paper, we present new algorithms
    which compute
    $\CC(w,k)$ in $O(n)$ total time for general ordered alphabets,
    and $\NN(w,k)$ in $O(n)$ total time for linearly sortable alphabets,
    for all positions $k = 1, \ldots, n$ in $w$.
    We further derive model-dependent optimal bounds by separating the algorithms into preprocessing and linear-time postprocessing: for $\CC$ the preprocessing is run reporting, and for $\NN$ it is preprocessing based on
    \sinote*{abbriviations}{%
      longest previous non-overlapping factors (LPnF) and longest next factors (LNF).
    }%
    In particular, all values $\CC(w,k)$ can be computed in $O(n\log n)$ time over general unordered alphabets in which direct accesses to alphabet characters
are restricted to equality tests, and in $O(n\log\sigma)$ time in the word RAM model, where $\sigma$ denotes the number of distinct characters occurring in $w$.
    For $\NN(w,k)$, the equality-testing complexity over general unordered alphabets is $\Theta(n^2)$.
    We also show that our upper bounds are optimal for all of the
    aforementioned alphabet assumptions and computation models. 
\end{abstract}


\maketitle
\hypersetup{pdfauthor={Haruki Umezaki, Hiroki Shibata, Dominik Koppl, Yuto Nakashima, Shunsuke Inenaga, Hideo Bannai}}

\section{Introduction} \label{sec:intro}

Let $w$ be a string of length $n$.
The problem of counting factors crossing a position was posed as
Problem 64 in Crochemore, Lecroq, and Rytter's textbook
``125 Problems in Text Algorithms''~\cite{125Problems}.
It asks to count the number $\CC(w,k)$ (resp. $\NN(w,k)$) of distinct
substrings in $w$ that have occurrences containing (resp. not containing)
a position $k$ in $w$.
According to the textbook~\cite{125Problems},
the notions of $\CC(w,k)$ and $\NN(w,k)$ are inspired by
the notion of \emph{string attractors}~\cite{KempaP18},
which form a set $\mathcal{P} = \{p_1, \ldots, p_\gamma\}$ of $\gamma$ positions
such that any substring of $w$ has an occurrence containing
a position $p_i \in \mathcal{P}$.
Besides this origin,
how efficiently one can compute $\CC(w,k)$ and $\NN(w,k)$ for a given string $w$, is an intriguing stringology question.

The solutions provided in the textbook~\cite{125Problems} compute
$\CC(w,k)$ and $\NN(w,k)$ in $O(n)$ time \emph{for a single position} $k$ in $w$
for constant-size alphabets.
Thus, a direct application of their solutions to the
\emph{all-position variant} of the problems,
which ask to compute $\CC(w,k)$ and $\NN(w,k)$ for \emph{all positions} $k = 1, \ldots, n$ in $w$, requires $O(n^2)$ total time.

In this paper, we present new algorithms
which compute for all positions $k = 1, \ldots, n$,
$\CC(w,k)$ in $O(n)$ total time and space for general ordered alphabets,
and $\NN(w,k)$ in $O(n)$ total time and space for linearly sortable alphabets.
Our solution for computing $\CC(w,k)$ for $k = 1, \ldots, n$
exploits the combinatorial property of the problem
and utilizes the \emph{runs} (a.k.a. \emph{maximal repetitions})~\cite{KolpakovK99}
occurring in $w$, which is completely different from the original solution from the textbook~\cite{125Problems}.

In this paper, we also make explicit how the complexity depends on the alphabet model.
For computing all values $\CC(w,k)$, the run-based framework separates the task into run reporting and a linear-time sweep: if all runs of $w$ are available in $\Truns(n)$ time, then all values $\CC(w,k)$ can be computed in $\Truns(n)+O(n)$ time.
Thus, linear-time run computation over general ordered alphabets gives the stated $O(n)$ bound as mentioned above, while over general unordered alphabets the classical Main--Lorentz framework~\cite{MainLorentz84} gives an equality-comparison-only $O(n\log n)$ bound and the word RAM algorithm of Ellert et al.~\cite{EllertGG23} gives an $O(n\log\sigma)$ bound,
where $\sigma$ denotes the number of distinct characters occurring in $w$.
The latter bound is tight in the same alphabet-sensitive sense: the classical $\Omega(n\log n)$ lower bound for square-freeness is the case $\sigma=\Theta(n)$, while Ellert et al. prove the matching $\Omega(n\log\sigma)$ lower-bound mechanism for strings with $\sigma$ distinct characters.
Similarly, for computing all values $\NN(w,k)$, the algorithm uses preprocessing tables for
\sinote*{abbriviations}{%
  \emph{longest previous non-overlapping factors} (\emph{LPnF})~\cite{crochemore11computing} and \emph{longest next factors} (\emph{LNF});
}%
  after these tables and the first value are available, the remaining update phase is linear. Denoting the time for this preprocessing by $\TLPnF(n)$, the total time is $\TLPnF(n)+O(n)$.
Thus, for general ordered alphabets,
we can compute all values $\NN(w,k)$ in $O(n \log n)$ time.
For general unordered alphabets with equality tests only, we can first rename the input characters in $O(n^2)$ time and then apply the linearly-sortable-alphabet algorithm; this is optimal by the quadratic lower bound for element distinctness in the equality-testing model.
We further prove that these bounds are optimal under the corresponding alphabet models, using square-freeness for $\CC$ and element distinctness for $\NN$.
Table~\ref{tab:models} summarizes the resulting alphabet-model dependence of the bounds discussed in this paper.

\begin{table}
\label{tab:models}
\caption{Alphabet-model dependence of the bounds, all of which are tight. For general unordered alphabets, direct access to alphabet characters is restricted to equality tests. Here, the ``word RAM'' model denotes the standard unit-cost word RAM model supporting $O(1)$-time bit operations on words of $\Omega(\log n)$ bits, whereas such word-level bit operations are not used in the ``unit-cost RAM'' model.}
\begin{tabular}{@{}lllll@{}}
\toprule
task & alphabet & upper bound & computation model for upper bound & lower-bound source \\
\midrule
all $\CC(w,k)$ & general ordered & $O(n)$ & unit-cost RAM & output size \\
all $\CC(w,k)$ & general unordered & $O(n\log n)$ & unit-cost RAM & square freeness \\
all $\CC(w,k)$ & general unordered & $O(n\log\sigma)$ & word RAM & square freeness \\
all $\NN(w,k)$ & linearly sortable & $O(n)$ & unit-cost RAM & output size \\
all $\NN(w,k)$ & general ordered & $O(n\log n)$ & unit-cost RAM & element distinctness \\
all $\NN(w,k)$ & general unordered & $O(n^2)$ & unit-cost RAM & element distinctness \\
\bottomrule
\end{tabular}
\end{table}

A preliminary version of this paper appeared in~\cite{UmezakiSKNIB25}.
The new materials in this full version are tight lower bounds
for computing all $\CC(w, k)$ and $\NN(w, k)$ in different alphabets and computation models.
In addition, the descriptions and proofs of the algorithms for computing $\CC$ and $\NN$
have been substantially revised throughout for improved readability.

\section{Preliminaries} \label{sec:preliminaries}

\subsection{Strings}
Let $\Sigma$ be an alphabet.
An element of $\Sigma^*$ is called a \emph{string}.
The length of a string $w \in \Sigma^*$ is denoted by $|w|$.
The \emph{empty string} $\varepsilon$ is the string of length $0$.
Let $\Sigma^+ = \Sigma^* \setminus \{\varepsilon\}$.
For string $w = xyz$, $x$, $y$, and $z$ are called
a \emph{prefix}, \emph{substring}, and \emph{suffix} of $w$,
respectively.
Let $\Substr(w)$ and $\Suffix(w)$ denote
the sets of  substrings and suffixes of $w$, respectively.
For a string $w$ of length $n$, $w[i]$ denotes the $i$th character of $w$
and $w[i..j] = w[i] \cdots w[j]$ denotes the substring of $w$
that begins at position $i$ and ends at position $j$ for $1 \leq i \leq j \leq n$.
For convenience, let $w[i..j] = \varepsilon$ for $i > j$.

\sinote*{Alphabets and models}{
When we discuss lower bounds and alphabet-sensitive upper bounds, we distinguish between linearly-sortable alphabets, general ordered alphabets, and general unordered alphabets.
In alphabet-sensitive bounds, $\sigma$ denotes the number of distinct characters occurring in the input string under consideration.
In particular, for general unordered alphabets, we use the equality-testing model:
the algorithms may compare two alphabet symbols only for equality, and do not
use an order, integer encoding, hashing, or alphabet-indexed arrays.  Apart
from such direct accesses to alphabet symbols, the algorithms are implemented
in the standard unit-cost RAM model without word-level bit operations; in
particular, arithmetic on text positions, lengths, periods, counters, and array
indices is allowed.
}

For two non-empty strings $s$ and $w$,
let $\occ(s, w) = \{i \mid w[i..i+|s|-1] = s\}$ denote the set of occurrences of $s$ in $w$, where we identify an occurrence of $s$ with its starting position.
For each position $1 \leq k \leq |w|$ in $w$,
let
\begin{eqnarray*}
    \cocc_k(s, w) & = & \{i \in \occ(s, w) \mid i \leq k \leq i+|s|-1\} \\
    \ncocc_k(s, w) & = & \{i \in \occ(s, w) \mid i+|s|-1 < k \mbox{ or } k < i\}
\end{eqnarray*}
denote the sets of occurrences of string $s$ that cross (resp. do not cross) the position $k$ in $w$.
Let
\begin{eqnarray*}
    \C(w,k) & = & \{s \in \Sigma^+ \mid \cocc_k(s, w) \neq \emptyset\}, \\
    \N(w,k) & = & \{s \in \Sigma^+ \mid \ncocc_k(s, w) \neq \emptyset\} = \Substr(w[1..k-1]) \cup \Substr(w[k+1..|w|])
\end{eqnarray*}
denote the sets of substrings $s$ of string $w$ that have
crossing (resp. non-crossing) occurrence(s) for the position $k$ in $w$,
respectively.

\begin{problem}[Counting distinct substrings with (non-)crossing occurrences]
Given a string $w$ of length $n$,
compute $\CC(w, k) = |\C(w,k)|$ and $\NN(w, k) = |\N(w,k)|$ for all positions $k = 1, \ldots, n$ in $w$.
\end{problem}

\subsection{Repetitions and runs}
For a string $s$, an integer $p~(1 \leq p \leq |s|)$ is a period of $s$ if
$s[i]=s[i+p]$ for all $1 \leq i \leq |s|-p$.
The {\em exponent} of $s$ is the rational $|s|/p$,
where $p$ is the smallest period of $s$.
A string $s \in \Sigma^+$ is said to be \emph{periodic} if
the exponent of $s$ is at least~$2$, or equivalently,
$s$'s smallest period is at most $|s|/2$.
A string is \emph{square-free} if it contains no substring of the form $xx$ with $x \in \Sigma^+$.
A maximal periodic substring $s = w[i..j]$ of $w$,
i.e., the smallest period $p$ of $s$
does not extend to the left of position $i$ nor to the right of position $j$,
namely,
\[
(i=1 \mbox{ or } w[i-1] \neq w[i+p-1])
\quad\mbox{and}\quad
(j=|w| \mbox{ or } w[j+1] \neq w[j-p+1]),
\]
is called a \emph{maximal repetition}, or \emph{run}, in $w$.
We identify a run $w[i..j]$ with the smallest period $p$
by a tuple $\langle i, j, p \rangle$.
Let $\Runs(w) = \{\langle i, j, p \rangle \mid \mbox{$w[i..j]$ is a run in $w$}\}$ denote the set of runs in $w$.

\begin{theorem}[\cite{BannaiIINTT17}] \label{theo:runs_theorem}
    $|\Runs(w)| < n$ holds for any string $w$ of length $n$.
\end{theorem}

\sinote*{added $O(n \log \sigma)$}{
\begin{theorem}[\cite{Ellert021,EllertGG23}] \label{theo:runs_linear_time}
  $\Runs(w)$ can be computed in $O(n)$ time for any length-$n$ string $w$ over a general ordered alphabet,
  and in $O(n \log \sigma)$ time for any length-$n$ string $w$ over a general unordered alphabet in the word RAM model,
  where $\sigma$ is the number of distinct characters in $w$.
\end{theorem}
}

\subsection{Suffix trees}
The \emph{suffix tree}~\cite{Weiner73} of a string $w$,
denoted $\STree(w)$, is a path-compressed trie representing $\Suffix(w)$ such that
(1) each internal node has at least two children,
(2) each edge is labeled by a non-empty substring of $w$, and
(3) the labels of out-going edges of the same node begin with distinct characters.
Each leaf of $\STree(w)$ is associated with the occurrence
of its corresponding suffix of $w$.

For a node $v$ of $\STree(w)$, let $\str(v)$ denote the string label
of the path from the root to $v$.
Each node $v$ stores its string depth $|\str(v)|$.
The \emph{locus} of a substring $s \in \Substr(w)$ in $\STree(w)$
is the position where $s$ is spelled out from the root.
The number of nodes in $\STree(w)$ is at most $2n-1$, where $n = |w|$.
We can represent $\STree(w)$ in $O(n)$ space
by representing each edge label $s$ with a pair $(i,j)$
of positions in $w$ such that $w[i..j] = s$.

Suppose that string $w$ terminates with an end-marker $\$$
that does not occur anywhere else in $w$.
Then, since $|\occ(y, w)| = 1$ holds for every suffix $y$ of $w$,
$\STree(w)$ has exactly $|w|$ leaves.
For $1 \leq i \leq |w|$,
$i$ is the suffix number of the leaf that represents $w[i..|w|]$.

For an alphabet model, let $\TST(n)$ denote the time needed to construct the suffix tree of a string of length $n$, with the standard auxiliary information used in this paper, such as string depths and leaves accessible in lexicographical order and corresponding suffix number order.
We use the following known model-dependent bounds.
\begin{fact}[Suffix-tree construction~\cite{Farach-ColtonFM00}]\label{the:stree_offline}
    For any string $w[1..n]$, $\STree(w)$ can be built in $\TST(n)$ time and $O(n)$ space under the alphabet model under consideration.
    In particular, $\TST(n)=O(n)$ over linearly-sortable alphabets.
    In the comparison model over general ordered alphabets, $\TST(n)=O(n\log n)$ by the sorting-complexity equivalence of suffix-tree construction.
\end{fact}

\sinote*{moved from N-section}{%
\subsection{Longest previous non-overlapping factors (LPnF) and longest next factors (LNF)}\label{sec:lpnf-lnf}

For a string $w$ of length $n$, the \emph{longest previous non-overlapping factor} (\emph{LPnF}) table of $w$ is an integer array $\LPNF_w[1..n]$
such that $\LPNF_w[i]$ stores the length of the longest prefix of $w[i..n]$
that has an occurrence in $w[1..i-1]$ for $1 \leq i \leq n$.
The \emph{longest next factor} (\emph{LNF}) table of $w$ is an integer array $\LNF_w[1..n]$ such that $\LNF_w[i]$ stores the length of the longest prefix of $w[i..n]$
that has an occurrence in $w[i+1..n]$ for $1 \leq i \leq n$.

\begin{lemma}[\cite{crochemore08lpf}, \cite{crochemore11computing}]\label{lemLPTF}
    Given a suffix-tree representation of $w$ supporting the standard bottom-up traversals used in LPF-type computations, the table $\LPNF_w$ can be built in $O(n)$ additional time and space.
    Consequently, if suffix-tree-type preprocessing for $w$ takes $\TST(n)$ time in the alphabet model under consideration, then $\LPNF_w$ can be built in $\TST(n)+O(n)$ time.
\end{lemma}

\begin{lemma}\label{lem:LNF-linear}
    Given the suffix tree $\STree(w)$ with its leaves accessible in suffix-number order, the table $\LNF_w$ can be built in $O(n)$ additional time and space.
    Consequently, if suffix-tree construction for $w$ takes $\TST(n)$ time in the alphabet model under consideration, then $\LNF_w$ can be built in $\TST(n)+O(n)$ time.
\end{lemma}

\begin{proof}
    Process the leaves of $\STree(w)$ in ascending order of their suffix numbers, while maintaining the compacted trie induced by the suffixes that have not yet been processed.
    Thus, just before the leaf for suffix $w[i..n]$ is processed, the maintained tree represents exactly the suffixes $w[i..n],w[i+1..n],\ldots,w[n..n]$.
    The longest prefix of $w[i..n]$ that occurs in one of the later suffixes $w[i+1..n],\ldots,w[n..n]$ is therefore the string spelled out by the parent of this leaf in the maintained compacted trie.
    We write this string depth into $\LNF_w[i]$.
    After that, we remove the leaf for suffix $i$.
    If its parent becomes a unary internal node, we contract this node with its remaining child, preserving the compacted-trie representation of the remaining suffixes.
    Each leaf is removed once, and each internal node and edge is charged only constantly many times by these removals and contractions.
    Hence all values of $\LNF_w$ are computed in $O(n)$ additional time and space once the suffix tree is available.
\end{proof}
}%

\section{Computing $\CC(w,k)$ for all positions $k$ in a string $w$}

In this section, we show how to compute $\CC(w,k)$ in $O(n)$ total time
for all positions $k$ in a given string $w$ of length $n$ over an ordered alphabet.

In our algorithm for computing $\CC(w, k)$,
we first compute the size of the multiset of substrings that cross position $k$ in $w$,
and then subtract the number $\DD(w,k)$ of duplicates.
Let $\U(w,k)$ be the multiset of substrings crossing $k$ in a given string $w$.
Since $|\U(w,k)|$ is equal to the number of intervals including $k$ in~$w$,
$|\U(w,k)| = k(|w|-k+1)$ holds: $[i,j]$ includes $k$ iff
$i\in [1,k]$ and $j\in[k,|w|]$.
We write
\[
  \DD(w,k)=|\U(w,k)|-\CC(w,k)
  =\sum_{x\in\C(w,k)}(|\cocc_k(x,w)|-1),
\]
namely, the number of duplicate crossing occurrences in the multiset $\U(w,k)$.

Let us consider how to compute $\DD(w,k)$.
The following observation and lemma are a key.
\begin{observation} \label{observation1}
    For any substring $x$ and position $k$ in string $w$,
    if $|\cocc_k(x, w)| \geq 2$, then $x$ is a substring of a run of $w$
    with smallest period $p < |x|$.
\end{observation}

We use the following well-known result:
\begin{lemma}[Weak periodicity lemma~\cite{Fine1965}]\label{lem:fine-wilf}
  If a string of length at least $p+q$ has periods $p$ and $q$,
  then $\gcd(p, q)$ is also a period of the string.
\end{lemma}

\begin{lemma}\label{lemma:equidistant_occurrences}
    For a run $r = \langle i,j,p\rangle$ of a string $w$,
    the distance $d$ between any two consecutive occurrences of a substring $x$ in $r$ with $|x|\geq p$
    must be $p$.
\end{lemma}
\begin{proof}
    \sinote*{Reworded}{
    Since $r$ has period $p$, shifting an occurrence of $x$ by $p$ inside the run gives another occurrence whenever the shifted interval is still contained in $r$.
    Hence the distance $d$ between two consecutive occurrences in $r$ is at most $p$.
    Suppose, for a contradiction, that $d<p$.
    Let the two consecutive occurrences start at positions $s$ and $s+d$.
    Then the substring $y = w[s..s+d+|x|-1]$ has period $d$, and it also has period $p$ because it is contained in the run $r$.
    Its length is $d+|x|\geq d+p$.
    By Lemma~\ref{lem:fine-wilf}, $\gcd(d,p)<p$ is a period of this substring $y$.
    Since $\gcd(d,p)$ divides $p$, this smaller period propagates along the $p$-periodic run and becomes a period of $r$, contradicting the minimality of $p$.
    Thus $d=p$.
    }
\end{proof}
Let $\Runs(w,k) = \{\langle i, j, p \rangle \in \Runs(w) \mid i \leq k \leq j\}$ denote the set of runs in $w$ that cross position $k$.
\sinote*{Reworded}{
For a run $r=\langle i,j,p\rangle\in\Runs(w,k)$, we charge duplicates by pairs of consecutive crossing occurrences at distance $p$ inside $r$.
More precisely, let
\[
  P(r,k)=\{(g,h) \mid i\leq g\leq k-p,
      \ k\leq h\leq j-p\}.
\]
      
For each $(g,h)\in P(r,k)$, the two occurrences
$w[g..h]$ and $w[g+p..h+p]$ are equal because $r$ has period $p$, and both cross $k$.
We define
\[
  \dup(r,k)=|P(r,k)|.
\]
}
Equivalently,
\begin{numcases}{\dup(\langle i,j,p \rangle, k)=}
    0                           & \text{if $i \leq k \leq i+p-1$}, \\
    \text{$(k-i-p+1)(j-p+1-k)$} & \text{if $i+p \leq k \leq j-p$},\label{dupequation} \\
    0                           & \text{if $j-p+1 \leq k \leq j$}.
\end{numcases}

\begin{figure}[pos=htbp]
    \centering
    \includegraphics[scale=0.55]{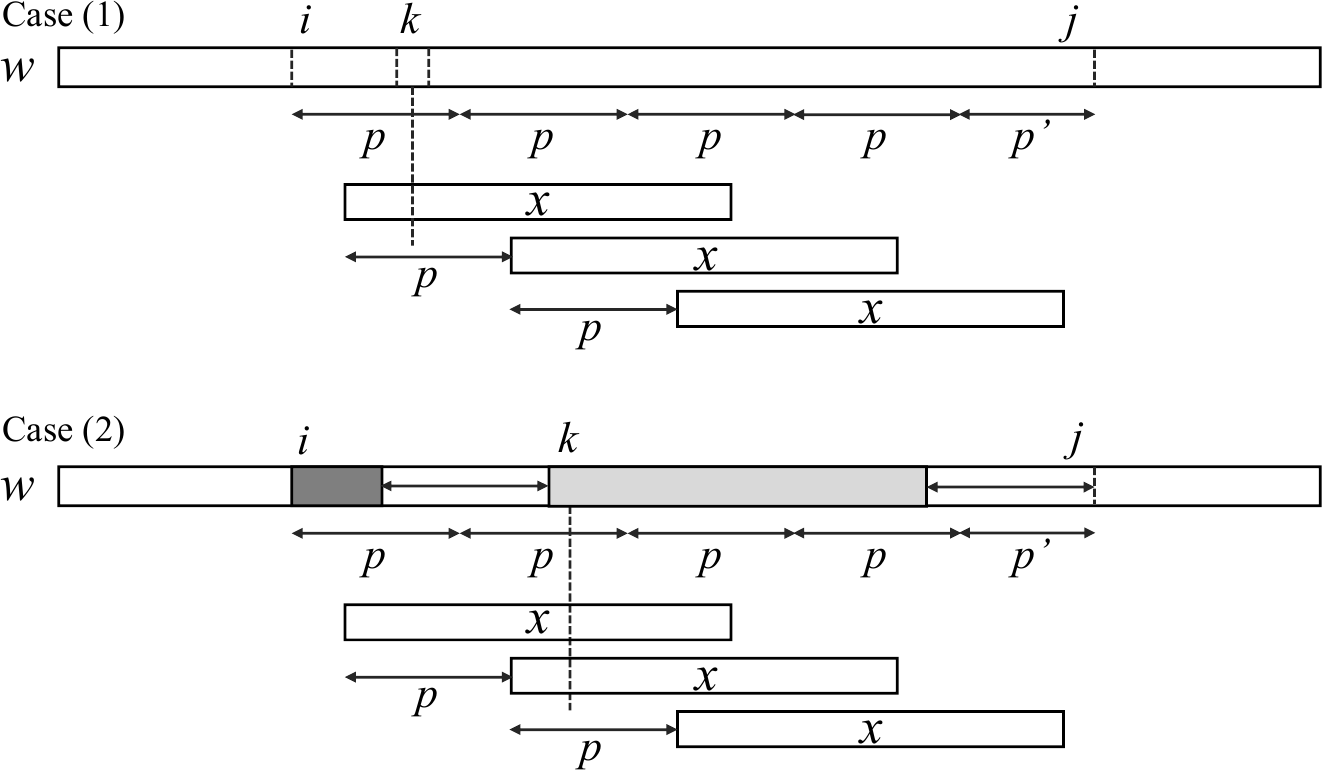}
    \caption{Illustration for $\dup(\langle i, j, p \rangle, k)$ for Cases (1) and (2). }
    \label{fig:dup}
\end{figure}

See Fig.~\ref{fig:dup}.
In Cases (1) and (3), there is no pair $(g,h)$ satisfying the inequalities in the definition of $P(r,k)$, and hence $\dup(r,k)=0$.
In Case (2), the start position $g$ can be chosen from the dark gray region of length $k-i-p+1$, and the end position $h$ can be chosen from the light gray region of length $j-p+1-k$.
For every such pair, the two occurrences $w[g..h]$ and $w[g+p..h+p]$ both cross $k$.
Thus the number of charged pairs is $(k-i-p+1)(j-p+1-k)$.

\begin{lemma} \label{lemma_numseq}
    $\DD(w,k) = \sum_{\langle i, j, p \rangle \in \Runs(w,k)} \dup(\langle i,j,p \rangle, k)$.
\end{lemma}

\begin{proof}
    \sinote*{More detailed proof}{
    For a fixed substring $x$ with crossing occurrences
    $a_1<a_2<\cdots<a_t$, where $a_q\in\cocc_k(x,w)$,
    the contribution of $x$ to $\DD(w,k)$ is $t-1$.
    We regard this contribution as the $t-1$ \emph{duplicate units}
    $(x,a_q,a_{q+1})$ for $1\leq q<t$.
    It is enough to show that the pairs counted by the sets $P(r,k)$, over all runs $r\in\Runs(w,k)$, are in one-to-one correspondence with these duplicate units.

    First take a pair $(g,h)\in P(r,k)$ for a run $r=\langle i,j,p\rangle$, and let $x=w[g..h]$.
    Since $r$ has period $p$, $x$ has another occurrence starting at $g+p$.
    By the definition of $P(r,k)$, both occurrences starting at $g$ and $g+p$ cross $k$.
    Moreover, there is no crossing occurrence of $x$ whose start lies strictly between $g$ and $g+p$.
    Indeed, any such occurrence would be fully contained in $w[i..j]$, because its start is larger than $g\geq i$ and its end is smaller than $h+p\leq j$.
    Then, among the occurrences of $x$ inside the run $r$ between
positions $g$ and $g+p$, there must be two consecutive occurrences
whose starting positions differ by less than $p$, contradicting Lemma~\ref{lemma:equidistant_occurrences}.
    Thus $(x,g,g+p)$ is a duplicate unit.

    Conversely, consider a duplicate unit $(x,a,b)$, where $a<b$ are consecutive crossing occurrence starts of $x$, and let $d=b-a$ and $|x|=\ell$.
    Since both occurrences cross the same position $k$, we have $d<\ell$.
    The factor $w[a..b+\ell-1]$ has period $d$ and length $d+\ell\geq 2d+1$.
    Extend this factor maximally to the left and to the right while preserving period $d$, and let the resulting interval be $[i,j]$.
    We claim that the shortest period of $w[i..j]$ is exactly $d$.
    If it had a smaller period $q<d$, then the occurrence of $x$ starting at $a+q$ would also be contained in $w[i..j]$: its start is between $a$ and $b$, and its end is at most $a+q+\ell-1 < b+\ell-1\leq j$.
    Since both intervals $[a,a+\ell-1]$ and $[a+q,a+q+\ell-1]$ are contained in $[i,j]$ and $w[i..j]$ has period $q$, the two factors are equal, so this is indeed another occurrence of $x$.
Equivalently, the smaller period $q$ propagates through the whole $d$-periodic interval, because all positions compared above remain inside the maximal interval $[i,j]$.
    This occurrence also crosses $k$, because $a+q<b\leq k$ and $a+q+\ell-1\geq a+\ell-1\geq k$.
    This contradicts the choice of $b$ as the next crossing occurrence of $x$ after $a$.
    Hence $w[i..j]$ is a run $r=\langle i,j,d\rangle$.
    Finally, if $h=a+\ell-1$, then $i\leq a\leq k-d$ since $b=a+d\leq k$, and $k\leq h\leq j-d$ since the shifted occurrence ending at $h+d$ is contained in the run.
    Therefore $(a,h)\in P(r,k)$, and this pair is mapped back to the duplicate unit $(x,a,b)$.

    The two constructions are inverse to each other, and hence each duplicate unit is counted exactly once by the right-hand side.
    The equality follows.
    }
\end{proof}

After $O(n)$-time preprocessing for computing $\Runs(w)$ with Theorem~\ref{theo:runs_linear_time},
Observation~\ref{observation1} and Lemma~\ref{lemma_numseq}
immediately lead us to an $O(n)$-time solution to compute $\CC(w,k)$
for a \emph{fixed} $k$.

Our strategy to compute $\CC(w,k)$ for all $k= 1, \ldots, n$ is
first to compute $\CC(w,1) = |\U(w,1)| - \DD(w,1)$ for $k = 1$ in $O(n)$ time,
and compute $\CC(w,k) = |\U(w,k)|-\DD(w,k)$ in amortized $O(1)$ time for increasing $k = 2, \ldots, n$.
Since $|\U(w,k)|$ is computable in $O(1)$ time by a simple arithmetic
for every $k$,
in what follows we focus on how to compute $\DD(w,k)$.

The next lemma exploits a useful structure of $\dup(\langle i, j, p \rangle, k)$ for the consecutive positions $k = i+p, \ldots, j-p$.
\begin{lemma}\label{lem:num_seq_diff}
  For a run $r=\langle i,j,p\rangle$, let
  $s_r=i+p$, $t_r=j-p$, and $b_r=t_r-s_r+1=j-i-2p+1$.
  Consider the sequence
  \[
    num_r=\dup(r,s_r),\dup(r,s_r+1),\ldots,\dup(r,t_r)
  \]
  of $b_r$ integers.
  Then $num_r[a]=a(b_r+1-a)$ for $a=1,\ldots,b_r$.
  In particular, for $1\leq a<b_r$,
  \[
    num_r[a+1]-num_r[a]=b_r-2a.
  \]
  Hence, if $b_r\geq 2$, the difference sequence of $num_r$ is the
  arithmetic progression
  \[
    b_r-2,\ b_r-4,\ \ldots,\ 2-b_r
  \]
  with first term $b_r-2$ and common difference $-2$.
\end{lemma}

\begin{proof}
    For $a=1,\ldots,b_r$, the $a$th term of $num_r$ corresponds to
    position $k=s_r+a-1=i+p+a-1$.  By Case~(\ref{dupequation}),
    \[
      \dup(r,k)=(k-i-p+1)(j-p+1-k)
              =a(b_r+1-a).
    \]
    Therefore, for $1\leq a<b_r$,
    \[
      num_r[a+1]-num_r[a]
      =(a+1)(b_r-a)-a(b_r+1-a)=b_r-2a.
    \]
    The claimed arithmetic progression follows immediately.
\end{proof}

\begin{figure}[pos=htbp]
    \centering
    \includegraphics[keepaspectratio,width=0.72\textwidth]{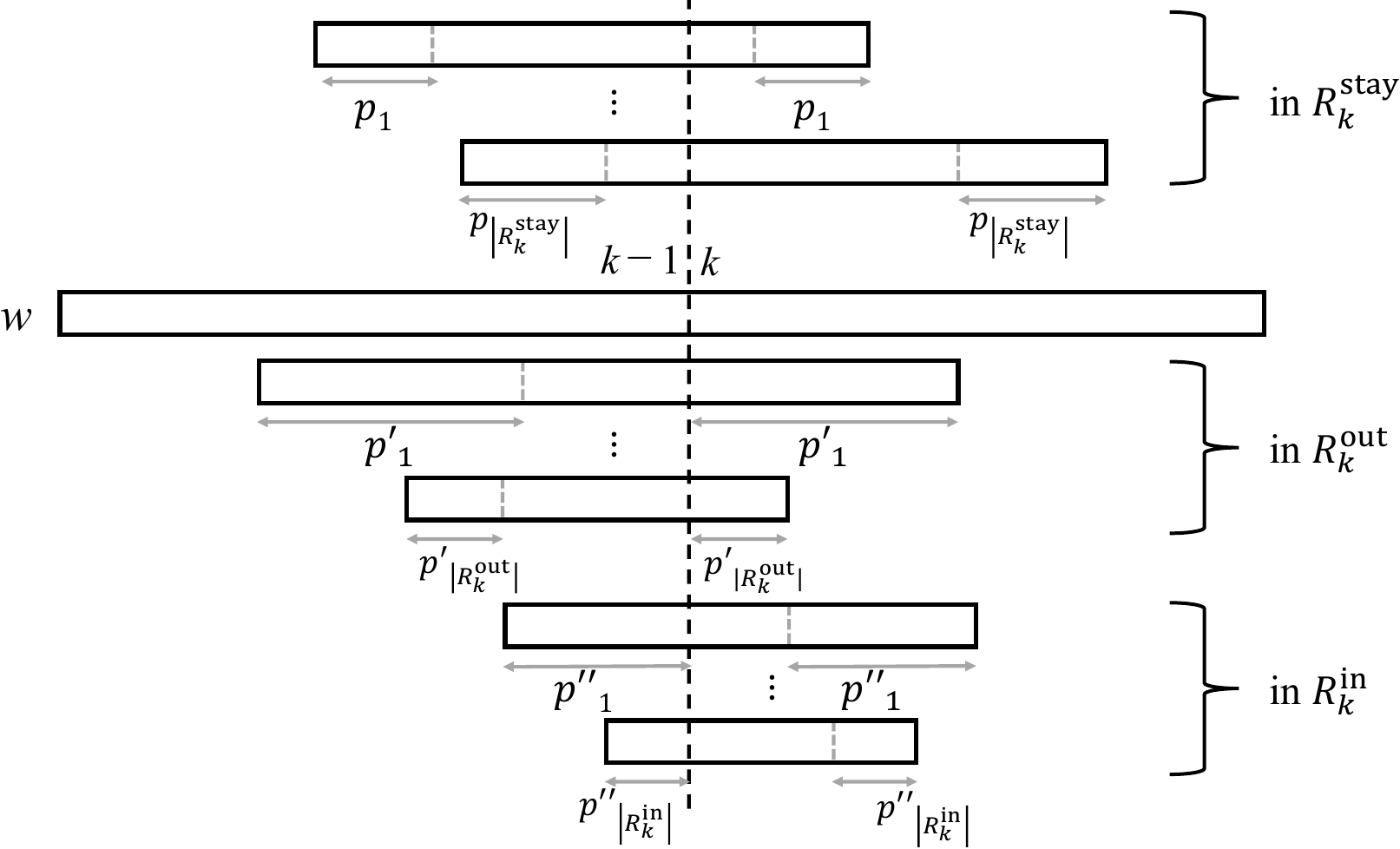}
    \caption{Illustration for the runs that stay, leave, and enter in a one-step transition.}
    \label{graph_algo}
\end{figure}

\sinote*{changed}{
For each position $k$, let
\[
  R_k = \{\langle i,j,p\rangle \in \Runs(w,k) \mid i+p\leq k\leq j-p\}
\]
be the set of runs that contribute positively to $\DD(w,k)$, and put $R_0=\emptyset$.
Equivalently, for a run $r=\langle i,j,p\rangle$, if we write
$s_r=i+p$, $t_r=j-p$, and $b_r=t_r-s_r+1=j-i-2p+1$, then
$r\in R_k$ exactly when $k\in[s_r,t_r]$.
In the transition from $k-1$ to $k$, we partition $R_{k-1}\cup R_k$ into
\[
  R_k^{\mathrm{stay}} = R_{k-1}\cap R_k,\qquad
  R_k^{\mathrm{out}}  = R_{k-1}\setminus R_k,\qquad
  R_k^{\mathrm{in}}   = R_k\setminus R_{k-1}.
\]
Thus $R_k^{\mathrm{stay}}$ consists of the runs that contribute positively at both positions,
$R_k^{\mathrm{out}}$ consists of the runs whose run interval ends at position $k-1$,
and $R_k^{\mathrm{in}}$ consists of the runs whose run interval starts at position $k$;
see Fig.~\ref{graph_algo} for the analogous one-step decomposition.
Recall that for a run $r=\langle i,j,p\rangle$ and
$k\in[i+p,j-p]$, we have
$\dup(r,k)=(k-i-p+1)(j-p+1-k)$.  In the notation of
Lemma~\ref{lem:num_seq_diff}, this contribution sequence has first value
$b_r$ and, when $b_r\geq 2$, successive differences
$b_r-2,b_r-4,\ldots,2-b_r$.
}

By Lemma~\ref{lemma_numseq},
$\DD(w,k) = \sum_{r\in R_k} \dup(r,k)$ holds, since runs outside $R_k$ have zero contribution at $k$.
\sinote*{Revised recurrence.}{
It is convenient to express the sweep by the above stay/out/in decomposition.
Let
\[
  f_k=\sum_{r\in R_k} b_r,\qquad
  d_k=\sum_{r\in R_k}(k-s_r),
\]
and let $e_k=\sum b_r$, where the sum is over runs whose run interval ends at position $k$.
Then, with $e_0=0$, the values $\DD(w,k)$ satisfy
\[
  \DD(w,k)=\DD(w,k-1)+f_k-e_{k-1}-2d_k\qquad (1\leq k\leq n),
\]
where $\DD(w,0)=0$.
Indeed, each run in $R_k^{\mathrm{stay}}$ changes its contribution by
$b_r-2(k-s_r)$ by Lemma~\ref{lem:num_seq_diff}; each run in
$R_k^{\mathrm{in}}$ contributes its first value $b_r$ at position $k$; and each run in
$R_k^{\mathrm{out}}$ contributed its last value $b_r$ to $\DD(w,k-1)$ and must be subtracted.
Since $R_k=R_k^{\mathrm{stay}}\cup R_k^{\mathrm{in}}$ and every run in $R_k^{\mathrm{in}}$ has $s_r=k$,
we have
\[
  \sum_{r\in R_k^{\mathrm{stay}}}\!\bigl(b_r-2(k-s_r)\bigr)
  + \sum_{r\in R_k^{\mathrm{in}}}\! b_r
  = f_k-2d_k.
\]
Also, $\sum_{r\in R_k^{\mathrm{out}}} b_r=e_{k-1}$.
This proves the recurrence.
We now show how to maintain $m_k=|R_k|$, $f_k$, and $d_k$ while scanning $k=1,\ldots,n$.


\begin{algorithm2e}[ht]
     \KwIn{a string $w[1..n]$ and the set $\Runs(w)$}
     \KwOut{$\CC(w,k)$ for all $k=1,2,...,n$}
     \SetAlgoLined
     \label{C01} Initialize position-indexed lists $\mathsf{Beg}[1..n]$ and $\mathsf{End}[0..n]$\;
     \label{C1}  \ForEach{run $\langle i,j,p\rangle\in\Runs(w)$}{
       $b\leftarrow j-i-2p+1$\;
       append $(\langle i,j,p\rangle,b)$ to $\mathsf{Beg}[i+p]$ and to $\mathsf{End}[j-p]$\;
     }
     \label{C03}  $m \leftarrow 0, d \leftarrow 0, f \leftarrow 0, \DD(w,0) \leftarrow 0$\;
     \label{C04}  \For{all $k=1,\ldots,n$}{
                      $e \leftarrow 0$\;
                      $d \leftarrow d+m$\;
     \label{C05}    \ForEach{$(\langle i,j,p\rangle,b)\in\mathsf{End}[k-1]$}{
     \label{C06}      $f \leftarrow f-b$\;
                      $m \leftarrow m-1$\;
                      $d \leftarrow d-b$\;
                      $e \leftarrow e+b$\;
                      }
     \label{C07}    \ForEach{$(\langle i,j,p\rangle,b)\in\mathsf{Beg}[k]$}{
     \label{C08}      $f \leftarrow f+b$\;
     \label{C09}      $m \leftarrow m+1$\;
     \label{C11}      }
     \label{C19}     $\DD(w,k) \leftarrow \DD(w,k-1) + f - e - 2d$\;
     \label{C20}     $\CC(w,k) \leftarrow k(n-k+1) - \DD(w,k)$\;
     \label{C21}     }
     \caption{Compute $\CC(w,k)$ for all positions}
     \label{alg:C}
\end{algorithm2e}

A pseudo-code of the proposed algorithm is shown in Algorithm~\ref{alg:C}.
Below, we describe our algorithm.

For each run $r=\langle i,j,p \rangle$, we call the interval $[i+p, j-p]$ the run interval for $r$ and write its length as $b_r=j-i-2p+1$.
To find runs by the starting and ending positions of their run intervals, we use two arrays indexed by positions.
For each run $r=\langle i,j,p\rangle$, we append $r$ to $\mathsf{Beg}[i+p]$ and to $\mathsf{End}[j-p]$.
These arrays allow us to access all runs whose run intervals start or end at the current position when we process the string positions $k=1,\ldots,n$ in increasing order.
Since each run is inserted into exactly one cell of $\mathsf{Beg}$ and one cell of $\mathsf{End}$, the arrays are built in $O(n)$ time and space after the runs are available.
No sorting of the runs is needed.

At the beginning of the iteration for position $k$, the maintained quantities refer to the runs in $R_{k-1}$, namely, the runs that contributed positively at position $k-1$.
We first increase $d$ by $m$, which changes each maintained distance from $k-1-s_r$ to $k-s_r$ for those runs.
We then remove the runs whose run interval ended at $k-1$, subtracting their values from $f$, $m$, and $d$, and accumulating their last values in $e$.
After that we insert the runs whose run interval starts at $k$.
The maintained variables are then exactly $f_k$, $d_k$, and $m_k$, and the recurrence above gives $\DD(w,k)$.

The discussion in this section can be summarized independently of the alphabet assumption, as follows:
\begin{theorem}\label{theo:C_from_runs}
Given the set $\Runs(w)$ of all runs of a string $w$ of length $n$, all values $\CC(w,k)$ for $k=1,\ldots,n$ can be computed in $O(n)$ additional time and space.
\end{theorem}
\begin{proof}
By Theorem~\ref{theo:runs_theorem}, the number of runs is less than $n$.
The algorithm above only scans the run intervals and the positions of $w$ a constant number of times, and computes $\DD(w,k)$ from the maintained quantities $m_k,f_k,d_k,e_k$.
Since $|\U(w,k)|=k(n-k+1)$ is obtained in constant time for every $k$, the claimed bound follows.

Consequently, if $\Runs(w)$ can be computed in $\Truns(n)$ time under a given alphabet model, then all values $\CC(w,k)$ can be computed in $\Truns(n)+O(n)$ time.
Thus, 
all $\CC(w,k)$ can be computed in $O(n)$ time for general ordered alphabets,
and, writing $\sigma$ for the number of distinct characters occurring in $w$,
in $O(n\log\sigma)$ time for general unordered alphabets
in the word RAM model by Theorem~\ref{theo:runs_linear_time}.
The $O(n\log n)$-time equality-testing instantiation for general unordered
alphabets in the unit-cost RAM model is discussed in Section~\ref{sec:unordered-runs}.
}

\section{An equality-comparison upper bound for $\CC$ over unordered alphabets}\label{sec:unordered-runs}

This section gives the upper-bound ingredient for computing $\CC(w,k)$ over a general unordered alphabet in the equality-testing model.
The goal is only to report all runs; once the runs are available, Theorem~\ref{theo:C_from_runs} gives all values $\CC(w,k)$ in linear additional time.
Throughout this section, the only operation on alphabet characters is equality testing.
In particular, we do not construct suffix trees, suffix arrays, or LCE data structures, and we do not sort alphabet characters.

We use the classical divide-and-conquer framework of Main and Lorentz~\cite{MainLorentz84}.
For a recursion interval $I=[b,e]$, let $m=\lfloor(b+e)/2\rfloor$ be its split point.
A periodic interval $w[i..j]$ with period $p$ is called \emph{crossing at $I$} if $b\leq i\leq m<j\leq e$.
It is \emph{$I$-local maximal with respect to $p$} if it cannot be extended inside $I$ to the left or to the right while preserving period $p$.
For a fixed period $p$, such intervals are described by maximal blocks of positions $t$ satisfying $w[t]=w[t+p]$.
This gives the following form of the Main--Lorentz crossing computation, which is the one used by Algorithm~\ref{alg:ml-runs}.

\begin{lemma}[Main--Lorentz local-crossing computation]\label{lem:ml-crossing-representation}
\sinote{Lemma revised.}
For a recursion interval $I=[b,e]$ with split point $m$, the Main--Lorentz crossing computation can produce lists $\mathsf{Cand}_p(I)$, $1\leq p\leq \lfloor |I|/2\rfloor$, such that $\mathsf{Cand}_p(I)$ consists exactly of all triples $\langle i,j,p\rangle$ for which $w[i..j]$ is an $I$-local maximal crossing periodic interval of period $p$ and $j-i+1\geq 2p$.
The total size $\sum_p |\mathsf{Cand}_p(I)|$ is $O(|I|)$, and the lists are computed in $O(|I|)$ time using equality tests only.
Moreover, $|\mathsf{Cand}_p(I)|\leq 2$ for every fixed $p$.
\end{lemma}
\begin{proof}
Fix a period $p$.
For positions $t\in[b,e-p]$, call $t$ \emph{$p$-good} if $w[t]=w[t+p]$.
A maximal consecutive block $[a,c]$ of $p$-good positions induces exactly one $I$-local maximal $p$-periodic interval, namely $[a,c+p]$.
Conversely, if $[i,j]$ is an $I$-local maximal $p$-periodic interval, then $[i,j-p]$ is a maximal block of $p$-good positions.
Thus it is enough to characterize the relevant maximal $p$-good blocks.

The interval $[a,c+p]$ crosses the split iff $a\leq m<c+p$, which is equivalent to saying that the good block $[a,c]$ intersects the window
\[
        W_p=[m-p+1,m].
\]
Moreover, $|[a,c+p]|\geq 2p$ iff $|[a,c]|=c-a+1\geq p$.
Since $W_p$ has length $p$, any interval $[a,c]$ of length at least $p$ that intersects $W_p$ must contain at least one of the two endpoints $m-p+1$ and $m$ of $W_p$.
Indeed, if it contained neither endpoint, then it would be strictly contained in $W_p$ and would have length at most $p-1$.
Hence every candidate in $\mathsf{Cand}_p(I)$ is induced by the maximal $p$-good block containing one of the two anchor positions \(m-p+1\) and \(m\),
provided that the anchor is $p$-good.
There are therefore at most two candidates for each period $p$.
This proves $|\mathsf{Cand}_p(I)|\leq 2$, and hence $\sum_p |\mathsf{Cand}_p(I)|=O(|I|)$.

It remains to explain how the lists are produced in the Main--Lorentz crossing step.
For a $p$-good anchor $x\in\{m-p+1,m\}$, let
\begin{eqnarray*}
  \mathsf{Lext}_p(x) & = & \max\{\ell\geq 0 : b\leq x-\ell\text{ and }
       w[x-\ell..x] = w[x-\ell+p..x+p]\} \\
  \mathsf{Rext}_p(x) & = & \max\{\ell'\geq 0 : x+\ell'\leq e-p\text{ and }
  w[x..x+\ell']=w[x+p..x+p+\ell']\}.
\end{eqnarray*}
If $x$ is $p$-good, then the maximal $p$-good block containing $x$ is
\[
        [x-\mathsf{Lext}_p(x),\ x+\mathsf{Rext}_p(x)].
\]
Thus the candidate induced by this anchor, if any, is
\[
        \langle x-\mathsf{Lext}_p(x),\ x+\mathsf{Rext}_p(x)+p,\ p\rangle,
\]
and it is kept precisely when the good block has length at least $p$ and the resulting interval crosses the split.
If the two anchors induce the same block, we keep only one copy.

\sinote{Moved details to Appendix.}
The remaining algorithmic detail is the computation of the two maximal
$p$-good blocks for every period $p$. This primitive is exactly the crossing
step of Main and Lorentz, specialized here to equality tests and without using
suffix-tree or LCE data structures. For readability, we use this standard
computation here as an equality-comparison, array-scan primitive; Appendix~\ref{app:ml-crossing-details}
gives the concrete scans and explains how the left and right extension lengths $\mathsf{Lext}_p(x)$ and $\mathsf{Rext}_p(x)$ are obtained from a constant number of ordinary prefix-scan arrays.
Those scans use only character-equality tests and simple array accesses.
Consequently, all lists $\mathsf{Cand}_p(I)$ are produced in $O(|I|)$ time.
\end{proof}

We call the triples in the lists $\mathsf{Cand}_p(I)$ the \emph{$I$-local candidates} of the current recursion node.
A local candidate may fail to be a run for exactly two reasons: its declared period is not the shortest period of the interval, or the interval is not globally maximal in the whole input string.
Algorithm~\ref{alg:ml-runs} first removes all candidates whose declared period is not shortest by processing periods in increasing order and deleting candidates of multiple periods contained in already surviving candidates.
Only after this filtering step does it apply the two equality tests for global maximality.

\begin{algorithm2e}[tbp]
\caption{Main--Lorentz style run reporting over a general unordered alphabet}
\label{alg:ml-runs}
\KwIn{a string $w[1..n]$ over a general unordered alphabet}
\KwOut{all runs in $w$}
Add two fresh sentinels $w[0]=\#$ and $w[n+1]=\$$\;
\SetKwProg{Proc}{Procedure}{}{}
\SetKwFunction{Remove}{RemoveNonShortestPeriods}
\Proc{$\mathsf{Process}(b,e)$}{
  \If{$e-b+1<2$}{\Return}
  $m\leftarrow \lfloor(b+e)/2\rfloor$\;
  $P\leftarrow \lfloor(e-b+1)/2\rfloor$\;
  For each $p=1,\ldots,P$, construct $\mathsf{Cand}_p([b,e])$ from the maximal $p$-good blocks containing the anchors $m-p+1$ and $m$, as in Lemma~\ref{lem:ml-crossing-representation}\;
  \Remove{$\{\mathsf{Cand}_p([b,e])\}_{p=1}^P$}\;
  \For{$p\leftarrow 1$ \KwTo $P$}{
    \ForEach{unmarked $\langle i,j,p\rangle\in\mathsf{Cand}_p([b,e])$}{
      \If{$w[i-1]\neq w[i-1+p]$ and $w[j+1-p]\neq w[j+1]$}{
        output $\langle i,j,p\rangle$\;
      }
    }
  }
  $\mathsf{Process}(b,m)$\;
  $\mathsf{Process}(m+1,e)$\;
}
\Proc{\Remove{$\{\mathsf{Cand}_p\}_{p=1}^P$}}{
  Mark all candidates as alive\;
  \For{$q\leftarrow 1$ \KwTo $P$}{
    \ForEach{alive $c=\langle i,j,q\rangle\in\mathsf{Cand}_q$}{
      $h\leftarrow 2$\;
      \While{$hq\leq P$}{
        $\text{\normalfont\itshape deleted}\leftarrow \text{\normalfont false}$\;
        \ForEach{alive $c'=\langle i',j',hq\rangle\in\mathsf{Cand}_{hq}$}{
          \If{$i\leq i'$ and $j'\leq j$}{
            mark $c'$ as non-shortest-period\;
            $\text{\normalfont\itshape deleted}\leftarrow \text{\normalfont true}$\;
          }
        }
        \If{$\text{\normalfont\itshape deleted}=\text{\normalfont false}$}{\textbf{break}\;}
        $h\leftarrow h+1$\;
      }
    }
  }
}
$\mathsf{Process}(1,n)$\;
\end{algorithm2e}

\sinote{Filter proof restructured.}
The subroutine \textsc{RemoveNonShortestPeriods} is the shortest-period filter.
It uses only integer information: periods, endpoints, and interval containment.
We prove its correctness and running time in three steps.

\paragraph{Deletion is safe.}
Suppose that, while processing an alive candidate $c=\langle i,j,q\rangle$, the filter marks a candidate $c'=\langle i',j',hq\rangle$ because $i\leq i'$ and $j'\leq j$.
Since $w[i..j]$ has period $q$, the substring $w[i'..j']$ also has period $q$.
As $q<hq$, the declared period $hq$ of $c'$ is not the shortest period of $w[i'..j']$.
Thus every marked candidate is indeed a non-shortest-period candidate.

\paragraph{All non-shortest-period candidates are deleted.}
Let $c=\langle a,b,p\rangle$ be a local candidate whose shortest period is $q<p$.
Since $b-a+1\geq 2p$, the periodicity lemma implies $q\mid p$.
Let $c_q=\langle i,j,q\rangle$ be the $I$-local maximal $q$-periodic interval containing the interval $[a,b]$.
Then $c_q$ belongs to $\mathsf{Cand}_q(I)$ and covers $c$.
When period $q$ is processed, $c_q$ is still alive: if it had been marked earlier by some period $q'<q$, then $w[i..j]$, and hence also $w[a..b]$, would have period $q'$, contradicting the minimality of $q$ for $w[a..b]$.
It remains to see that the scan for multiples of $q$ starting from $c_q$ reaches the period $p$.
Let $p=kq$.
For every $h=2,\ldots,k-1$, the same interval $[i,j]$ is an $I$-local maximal candidate of period $hq$.
Indeed, $c_q$ has period $q$ and therefore period $hq$; also $j-i+1\geq b-a+1\geq 2kq\geq 2hq$, so the length requirement is satisfied.
If $c_q$ could be extended inside $I$ with period $hq$, then the $q$-periodicity inside $c_q$ would imply that the same extension also preserves period $q$, contradicting the $I$-local maximality of $c_q$ with respect to $q$.
Thus $\langle i,j,hq\rangle$ is present in $\mathsf{Cand}_{hq}(I)$.
Moreover, it is alive when $c_q$ is processed, because being marked by a smaller period would again give a period smaller than $q$ for $c_q$.
Consequently, the scan deletes an alive contained candidate for every multiple $2q,3q,\ldots,(k-1)q$, and hence it does not stop before reaching $p=kq$.
At period $p$, the candidate $c$ is contained in $c_q$ and is therefore marked.

\paragraph{Linear-time filtering.}
At a fixed node $I$, let $\mathcal{K}_I=\bigcup_p\mathsf{Cand}_p(I)$.
The period loop contributes $O(|I|)$ time.
Consider the inner scans that, for an alive candidate of period $q$, inspect the candidate lists for multiple periods $2q,3q,\ldots$.
Each successful scan iteration marks at least one previously alive candidate, and we charge the iteration to one candidate marked in that iteration.
No candidate is marked twice.
Each alive candidate $c$ causes at most one unsuccessful scan iteration, namely the first multiple at which no alive contained candidate is found and the scan for $c$ stops.
By Lemma~\ref{lem:ml-crossing-representation}, each fixed-period list has size at most two, so each scan iteration performs only constant many containment tests.
Therefore the filter spends $O(|I|+|\mathcal{K}_I|)=O(|I|)$ time at node $I$.

\begin{lemma}\label{lem:ml-filter-sound}
Every triple output by Algorithm~\ref{alg:ml-runs} is a run.
\end{lemma}
\begin{proof}
Let $\langle i,j,p\rangle$ be an output triple.
It is a periodic substring of length at least $2p$.
If $p$ were not the shortest period of $w[i..j]$, let $q<p$ be the shortest period.
Since $j-i+1\geq 2p$, the periodicity lemma implies $q\mid p$.
The $q$-periodic interval containing $w[i..j]$ and maximal inside the same recursion interval is an $I$-local candidate in $\mathsf{Cand}_q(I)$ covering $\langle i,j,p\rangle$.
By the correctness of \textsc{RemoveNonShortestPeriods} proved above,
the candidate $\langle i,j,p\rangle$ would have been marked when the
surviving $q$-periodic candidate covering it was processed, before the
final global maximality test. This is a contradiction.
Therefore $p$ is the shortest period of $w[i..j]$.
The final two equality tests state exactly that this shortest period cannot be extended to the left or to the right in the whole string.
Hence the output triple is a run.
\end{proof}

\begin{lemma}\label{lem:ml-filter-complete}
Every run of $w$ is output by Algorithm~\ref{alg:ml-runs}.
\end{lemma}
\begin{proof}
Let $\langle i,j,p\rangle$ be a run.
Consider the unique recursion node whose two children separate positions $i$ and $j$.
Its split point $m$ satisfies $i\leq m<j$, and hence the run crosses the split.
Since the run is globally maximal with respect to its shortest period $p$, it is also maximal inside this recursion interval.
Thus Lemma~\ref{lem:ml-crossing-representation} puts $\langle i,j,p\rangle$ in $\mathsf{Cand}_p(I)$.
It is not marked by \textsc{RemoveNonShortestPeriods} because $p$ is the shortest period, and it passes the final global maximality test.
Thus it is output.
\end{proof}

\begin{lemma}\label{lem:ml-no-duplicates}
No run is output more than once by Algorithm~\ref{alg:ml-runs}.
\end{lemma}
\begin{proof}
For a run $\langle i,j,p\rangle$, the recursion node whose two children separate $i$ and $j$ is unique.
At this owner node the full interval $[i,j]$ crosses the split.
At a proper descendant, the two endpoints are not both contained in a single recursion interval crossing its split.
At an ancestor or another node, the full interval is not the $I$-local maximal candidate owned by that node.
Hence the same full run cannot pass the global maximality test at two different nodes.
\end{proof}

\begin{theorem}\label{theo:ml-runs-unordered}
All runs of a string of length $n$ over a general unordered alphabet can be reported in $O(n\log n)$ time using equality comparisons only.
\end{theorem}
\begin{proof}
By Lemma~\ref{lem:ml-crossing-representation}, the local candidates at a recursion node $I$ are computed in $O(|I|)$ time using equality tests only.
The shortest-period filter is local to the node and is evaluated during the increasing-period scan of these lists; as noted above, the total number of local candidates at the node is $O(|I|)$ and each fixed-period list has constant size.
The final global maximality test uses two equality tests per surviving candidate.
Thus the work at node $I$ is $O(|I|)$.
The total length of all intervals at each recursion depth is $O(n)$, and the recursion depth is $O(\log n)$.
The total running time is therefore $O(n\log n)$.
The working space can be kept linear by processing the recursion depth-first: at any time, the prefix-scan arrays and candidate lists stored along the active recursion path have total size $O(n)$.
\end{proof}

\sinote{C consequence.}
Combining Theorem~\ref{theo:ml-runs-unordered} with Theorem~\ref{theo:C_from_runs}, all values $\CC(w,k)$ can be computed in $O(n\log n)$ time over a general unordered alphabet using equality comparisons only.

\section{Computing $\NN(w,k)$ for all positions $k$ in $w$}\label{sec:N}

In this section, we describe the part of the algorithm for computing $\NN$ that is independent of the alphabet model.
The only model-dependent ingredient is the preprocessing step that computes the LPnF and LNF tables defined in Section~\ref{sec:lpnf-lnf}.
By Lemmas~\ref{lemLPTF} and~\ref{lem:LNF-linear}, these tables can be constructed from suffix trees.
We also include in this preprocessing the computation of the initial value $\NN(w,1)$.
Once these data are available for $w$ and $w^R$, all remaining values $\NN(w,k)$ can be obtained by a single left-to-right scan.

We write $\TLPnF(n)$ for the time required, under the alphabet model under consideration, to compute the LPnF/LNF tables for $w$ and $w^R$ together with the initial value $\NN(w,1)$.
The main result of this section is therefore the following parameterized statement.
The alphabet-model dependent bounds then follow by substituting the corresponding value of $\TLPnF(n)$.

\begin{theorem}\label{theo:N_from_LPnF}
Given the tables $\LPNF_w,\LNF_w,\LPNF_{w^R},\LNF_{w^R}$ and the first value $\NN(w,1)$, all values $\NN(w,k)$, $1\leq k\leq n$, can be computed in $O(n)$ additional time and space.
Consequently, all values $\NN(w,k)$ can be computed in $\TLPnF(n)+O(n)$ time.
\end{theorem}

\begin{proof}
Let $A_x = \Substr(w[1..x])$ and $B_x=\Substr(w[x..n])$.
Then, $\NN(w,x) = |A_{x-1}\cup B_{x+1}|$.
The idea is to compute, for increasing values of $x$,
the two differences
$|A_x\cup B_{x+1}|-|A_{x-1}\cup B_{x+1}|$ and
$|A_x\cup B_{x+2}| - |A_{x}\cup B_{x+1}|$
so that $\NN(w,x+1) = |A_x\cup B_{x+2}|$ can be computed from $\NN(w,x)$ by adding these differences.
If we can find the two differences in constant time for each $x$,
then all values after $\NN(w,1)$ are obtained in linear total time.

We first claim that having the arrays $\LNF_w,\LPNF_w$ for $w$ at hand,
$|A_{x}\cup B_{x+2}| - |A_{x}\cup B_{x+1}|$
can be computed in $O(1)$ time.
Since $A_{x}\cup B_{x+2} \subseteq A_{x}\cup B_{x+1}$, we only need to count how many elements are removed,
which must be prefixes of $w[x+1..n]$.
The removed prefixes are the prefixes of $w[x+1..n]$
that do not occur in $A_{x}$ and do not occur in $B_{x+2}$.
\sinote*{reworded}{%
By the definitions of the LPnF and LNF tables,
$\alpha=\LPNF_w[x+1]$ is the maximum length $\ell$ such that
the prefix $w[x+1..x+\ell]$ has an occurrence in $w[1..x]$.
Thus, every prefix of $w[x+1..n]$ of length at most $\alpha$ belongs to $A_x$.
Similarly, $\beta=\LNF_w[x+1]$ is the maximum length $\ell'$ such that
the prefix $w[x+1..x+\ell']$ has another occurrence starting in $w[x+2..n]$.
Thus, every prefix of $w[x+1..n]$ of length at most $\beta$ belongs to
$B_{x+2}$.
}%
Therefore, the prefixes of $w[x+1..n]$ are removed if and only if they are longer than $\max(\alpha,\beta)$,
and their number is $n-x-\max(\alpha,\beta)$.

The case for $|A_{x}\cup B_{x+1}| - |A_{x-1}\cup B_{x+1}|$ is symmetric and can be computed using the arrays
$\LNF_{w^R}$ and $\LPNF_{w^R}$ for the reverse string $w^R$ in a similar fashion.

We now show the update formula explicitly.
Let $v=w^R$.  For $1\leq x<n$, define
\[
  \rho_x = n-x-\max\{\LPNF_w[x+1],\LNF_w[x+1]\},
\]
which is the number of substrings removed when $B_{x+1}$ is replaced by $B_{x+2}$ after $A_x$ has already been inserted.
Similarly, with $y=n-x+1$, define
\[
  \lambda_x = x-\max\{\LPNF_v[y],\LNF_v[y]\}.
\]
This is the number of suffixes of $w[1..x]$ that are newly inserted when $A_{x-1}$ is replaced by $A_x$.
Thus
\[
  \NN(w,x+1)=\NN(w,x)+\lambda_x-\rho_x.
\]
This completes the proof.
\end{proof}

\begin{algorithm2e}[H]
\caption{Computing all values $\NN(w,k)$ from LPnF/LNF tables}
\label{alg:all-N}
\KwIn{a string $w[1..n]$, the tables $\LPNF_w,\LNF_w,\LPNF_{w^R},\LNF_{w^R}$, and $N[1]=\NN(w,1)$}
\KwOut{$\NN(w,1),\ldots,\NN(w,n)$}
Let $v=w^R$\;
\For{$x\leftarrow 1$ \KwTo $n-1$}{
  $\rho\leftarrow n-x-\max\{\LPNF_w[x+1],\LNF_w[x+1]\}$\;
  $y\leftarrow n-x+1$\;
  $\lambda\leftarrow x-\max\{\LPNF_v[y],\LNF_v[y]\}$\;
  $N[x+1]\leftarrow N[x]+\lambda-\rho$\;
}
\Return{$N[1],\ldots,N[n]$}\;
\end{algorithm2e}

Algorithm~\ref{alg:all-N} shows the pseudo-code of our algorithm from Theorem~\ref{theo:N_from_LPnF}: after the preprocessing data are available, the algorithm performs constant work for each $x=1,\ldots,n-1$ and stores only linear-size arrays.

Theorem~\ref{theo:N_from_LPnF} immediately yields the following model-dependent consequences by plugging in the time needed for the LPnF/LNF preprocessing.

\begin{corollary}\label{theo:N-linearly-sortable}
If $\TLPnF(n)=O(n)$, then all values $\NN(w,k)$, $1\leq k\leq n$, can be computed in $O(n)$ time and space.
In particular, this holds over linearly-sortable alphabets.
\end{corollary}
\begin{proof}
By Fact~\ref{the:stree_offline}, $\TST(n)=O(n)$ over linearly-sortable alphabets.  Hence Lemmas~\ref{lemLPTF} and~\ref{lem:LNF-linear} imply that the required LPnF/LNF tables for both $w$ and $w^R$ are computable in linear time.
The first value $\NN(w,1)$ is the number of distinct substrings of $w[2..n]$ and is also obtained within the same suffix-tree-type preprocessing time.
Thus $\TLPnF(n)=O(n)$, and the claim follows from Theorem~\ref{theo:N_from_LPnF}.
\end{proof}

\begin{corollary}\label{theo:N-general-ordered}
Suppose that the LPnF/LNF preprocessing for $w$ and $w^R$, together with the first value $\NN(w,1)$, can be performed in $O(n\log n)$ time and $O(n)$ space.
Then all values $\NN(w,k)$, $1\leq k\leq n$, can be computed in $O(n\log n)$ time and $O(n)$ space.
\end{corollary}
\begin{proof}
Under the assumption, $\TLPnF(n)=O(n\log n)$.
The claim follows immediately from Theorem~\ref{theo:N_from_LPnF}.
\end{proof}

For strings over a general ordered alphabet, this assumption is satisfied in the comparison model: by Fact~\ref{the:stree_offline}, $\TST(n)=O(n\log n)$, and therefore the suffix-tree-type preprocessing used in Lemmas~\ref{lemLPTF} and~\ref{lem:LNF-linear}, for both $w$ and $w^R$, can be performed in $O(n\log n)$ time.
The first value $\NN(w,1)$ is computed within the same bound, and hence $\TLPnF(n)=O(n\log n)$ for this model.

We finally consider general unordered alphabets, where the only permitted operation on alphabet symbols is an equality test.
This is the standard equality-testing model for unordered alphabets; see also the model used in~\cite{EllertGG23}.
In contrast to the general ordered case, the LPnF/LNF preprocessing is not the most convenient primitive in this model.
Still, an immediate renaming step can be viewed as an $O(n^2)$-time preprocessing step after which the linearly-sortable-alphabet algorithm applies.

\begin{corollary}\label{coro:N-general-unordered-upper}
Over a general unordered alphabet in the equality-testing model, all values $\NN(w,k)$, $1\leq k\leq n$, can be computed in $O(n^2)$ time and $O(n)$ space.
\end{corollary}
\begin{proof}
Scan the string from left to right and maintain one representative for each distinct symbol seen so far.
When reading $w[i]$, compare it with the representatives until an equal representative is found, or create a new representative if none is equal.
Assign to $w[i]$ the integer name of the corresponding representative.
This renaming uses $O(n\sigma)\subseteq O(n^2)$ equality tests and transforms $w$ into a string over the linearly-sortable integer alphabet $[1..\sigma]$.
Hence, Corollary~\ref{theo:N-linearly-sortable} applies and gives all values $\NN(w,k)$ in an additional $O(n)$ time.
The representative list and the renamed string use linear space.
\end{proof}

\section{Lower bounds and optimality}\label{sec:optimality}

\sinote{New section.}
We now summarize the optimality of the preceding algorithms under the corresponding alphabet models.
The lower bounds in this section are worst-case lower bounds in the deterministic comparison/equality-testing decision-tree models.
\sinote{Parameterized connection.}
For $\CC$, Theorem~\ref{theo:C_from_runs} separates the computation into run reporting, taking $\Truns(n)$ time, and a linear-time sweep.
For $\NN$, Theorem~\ref{theo:N_from_LPnF} separates the computation into LPnF/LNF-type preprocessing, taking $\TLPnF(n)$ time, and a linear-time update phase.
Thus, after substituting the appropriate preprocessing bounds for the chosen alphabet model, the remaining work is always $O(n)$.
The lower-bound mechanisms are summarized in Table~\ref{tab:models}.
In the alphabet-sensitive statements below, $\sigma$ denotes the number of distinct characters occurring in the input string under consideration.

\subsection{Lower bound for $\CC$}

For a string $w$ of length $n$, let
\[
M(w,k)=k(n-k+1)
\]
be the number of intervals of $w$ containing position $k$.
Clearly $\CC(w,k)=M(w,k)$ for all $k$ iff no two equal substrings have occurrences that cross a common position.

\begin{lemma}\label{lem:square-to-C}
Given a string $x$, one can construct in linear time a string $\phi(x)$ such that $x$ is square-free iff
$\CC(\phi(x),k)=M(\phi(x),k)$ for every position $k$ of $\phi(x)$.
\end{lemma}
\begin{proof}
Let $x=x_1\cdots x_n$ and define
\[
\phi(x)=\#x_1\#x_2\#\cdots\#x_n\#,
\]
where $\#$ is a fresh character.
If $x[i..i+p-1]=x[i+p..i+2p-1]$, then the substring
$\#x_i\#\cdots\#x_{i+p-1}\#$ occurs in $\phi(x)$ at distance $2p$ with length $2p+1$.
The two occurrences overlap, and hence they cross a common position; therefore $\CC(\phi(x),k)<M(\phi(x),k)$ for some $k$.
Conversely, suppose that two equal substrings of $\phi(x)$ cross a common position.
Let their occurrences start at positions $a<a+d$ and have length $L$.
Since the two occurrences overlap, $L>d$.
Because $\phi(x)$ alternates the fresh marker with original symbols, the shift $d$ must be even; write $d=2p$.
The union of the two occurrences is a factor of $\phi(x)$ of length $L+d>2d=4p$ and has period $2p$.
Projecting this factor to the original-symbol positions yields two adjacent equal blocks of length $p$ in $x$.
Thus $x$ contains a square.
\end{proof}

The classical result of Main and Lorentz~\cite{MainLorentz84} gives an $O(n\log n)$ algorithm for detecting repetitions over a general alphabet and proves a matching $\Omega(n\log n)$ lower bound in the equality-testing model.
By Lemma~\ref{lem:square-to-C}, any algorithm computing all $\CC(w,k)$ over a general unordered alphabet can decide square-freeness.
Consequently, computing all values $\CC(w,k)$ requires $\Omega(n\log n)$ equality tests in the worst case for unrestricted general unordered alphabets, and the $O(n\log n)$ equality-comparison upper bound from Theorem~\ref{theo:ml-runs-unordered} is optimal in this worst-case sense.

The preceding $\Theta(n\log n)$ bound is the alphabet-insensitive worst-case statement.
Ellert, Gawrychowski, and Gourdel~\cite{EllertGG23} observed that the Main--Lorentz lower-bound instances use up to $n$ distinct characters, i.e., $\sigma=\Theta(n)$, and refined square-freeness testing over general unordered alphabets to a tight $\Theta(n\log\sigma)$ equality-comparison bound.
They also gave an $O(n\log\sigma)$-time word-RAM implementation and extended the upper bound to reporting all runs.
Combining this run-reporting algorithm with Theorem~\ref{theo:C_from_runs} gives an $O(n\log\sigma)$-time word-RAM algorithm for all values $\CC(w,k)$, while Lemma~\ref{lem:square-to-C} transfers their $\Omega(n\log\sigma)$ comparison lower bound to our problem.
Thus the $O(n\log\sigma)$ row in Table~\ref{tab:models} is tight and should be viewed as the alphabet-sensitive counterpart of the Main--Lorentz $\Theta(n\log n)$ worst-case bound.

\subsection{Lower bounds for $\NN$}

We use standard reductions from element distinctness.
First, in the comparison model over general ordered alphabets, element distinctness on $n$ elements requires $\Omega(n\log n)$ comparisons.
Let $a_1,\ldots,a_n$ be an instance of element distinctness and regard it as a string $w=a_1\cdots a_n$.
If all elements are distinct, then every substring of $w$ is unique, and hence
\[
\NN(w,k)=\frac{k(k-1)}{2}+\frac{(n-k)(n-k+1)}{2}
\]
for every $k$, because the non-crossing substrings are exactly the intervals strictly to the left or strictly to the right of $k$.
If two equal elements occur at positions $i<j$ and $n>2$, then there exists a position $k\notin\{i,j\}$.
For this $k$, the two length-one intervals spelling the same symbol are both non-crossing, and therefore the above value is too large by at least one.
Thus the array of all values $\NN(w,k)$ determines whether the input elements are pairwise distinct.
Consequently, computing all values $\NN(w,k)$ in the comparison model over general ordered alphabets is at least as hard as element distinctness.

\sinote{Moved lower bound.}
For the equality-testing model over general unordered alphabets, we also need the following elementary quadratic lower bound for element distinctness.
\begin{lemma}\label{lem:ed-equality-lower}
In the equality-testing model over a general unordered alphabet, deciding whether the $n$ input symbols are pairwise distinct requires $\Omega(n^2)$ equality tests in the worst case.
\end{lemma}
\begin{proof}
Consider any deterministic equality-testing algorithm and run it on an input on which every equality test is answered by ``not equal''.
If the algorithm performs fewer than $\binom{n}{2}$ tests on this execution, then there is a pair of positions $i<j$ that has never been compared.
The transcript is consistent with an input in which all $n$ symbols are pairwise distinct.
It is also consistent with an input in which the symbols at positions $i$ and $j$ are equal and all other symbols are distinct: every equality test seen by the algorithm still compares two positions other than the untested pair, or compares one of them with a third position, and hence can be answered ``not equal'' in this second input as well.
The algorithm receives the same transcript on a yes-instance and on a no-instance, so it cannot be correct on both.
Therefore every correct deterministic algorithm must, in the worst case, compare all $\binom{n}{2}$ pairs up to constant factors.
\end{proof}

The reduction above from element distinctness to the array of all values $\NN(w,k)$ is valid in the equality-testing model as well.
Together with Lemma~\ref{lem:ed-equality-lower}, it gives an $\Omega(n^2)$ lower bound for computing all values $\NN(w,k)$ over general unordered alphabets with equality tests only.

For $\NN$, Theorem~\ref{theo:N_from_LPnF} separates the computation into LPnF/LNF preprocessing, taking $\TLPnF(n)$ time, and a linear-time update phase.
Over linearly-sortable alphabets, $\TLPnF(n)=O(n)$, and hence the resulting $O(n)$ algorithm for $\NN$ is optimal by the trivial $\Omega(n)$ output-size lower bound.
In the comparison model over general ordered alphabets, the suffix-tree-based preprocessing gives $\TLPnF(n)=O(n\log n)$ by Corollary~\ref{theo:N-general-ordered}, while the element-distinctness reduction above gives the matching lower-bound mechanism.
For general unordered alphabets in the equality-testing model, Corollary~\ref{coro:N-general-unordered-upper} gives the matching $O(n^2)$ upper bound, and the preceding paragraph gives the $\Omega(n^2)$ lower bound.
Thus the bound is tight.

\section{Conclusions and discussion}\label{sec:conclusion}

We studied the all-position variant of the problem of counting distinct substrings that do or do not cross a given position.
For the crossing quantity $\CC(w,k)$, the computation naturally separates into two parts: first report the runs of the input string, and then accumulate their contributions to the duplicate crossing substrings by a linear-time sweep.
Thus, if all runs can be reported in $\Truns(n)$ time, then all values $\CC(w,k)$ are obtained in $\Truns(n)+O(n)$ time.
For the non-crossing quantity $\NN(w,k)$, the analogous separation is through LPnF/LNF preprocessing: once these tables and the first value are available, all remaining values are obtained by a linear-time update.
This gives a total running time of $\TLPnF(n)+O(n)$.

This separation also clarifies the role of the alphabet model and yields tight bounds in all cases considered in this paper.
For $\CC$, the resulting bounds are $O(n)$ time over general ordered alphabets, $O(n\log n)$ time over general unordered alphabets in the equality-testing model on a unit-cost RAM, and $O(n\log\sigma)$ time over general unordered alphabets in the word RAM model, where $\sigma$ is the number of distinct characters in the input string.
The corresponding lower bounds follow from square-freeness, including the alphabet-sensitive $\Omega(n\log\sigma)$ bound of Ellert et al.
For $\NN$, the bounds are $O(n)$ time over linearly sortable alphabets, $O(n\log n)$ time over general ordered alphabets, and $\Theta(n^2)$ time over general unordered alphabets with equality tests only.
These lower bounds are inherited from element distinctness.

The present paper considers the static all-position problem.
Natural variants include online, sliding-window, and fully dynamic versions, where all $\CC(w,k)$ and $\NN(w,k)$ have to be maintained as the input string changes.
The decompositions developed here suggest that the main challenges are dynamic
or sliding-window run reporting for $\CC$, and dynamic maintenance of the
non-overlapping factor information underlying the LPnF/LNF-based part of the
algorithm for $\NN$.

\section*{Acknowledgments}
This work was supported by JSPS KAKENHI Grant Numbers JP23K24808, JP23K18466~(SI), JP23H04378, JP25K21150 (DK), and JP24K02899 (HB).

\appendix
\clearpage
\appendix

\section{Details of the Main--Lorentz crossing computation}\label{app:ml-crossing-details}

This appendix gives the array-scan details used in Lemma~\ref{lem:ml-crossing-representation}.
The purpose is to make explicit that the local candidates in Section~\ref{sec:unordered-runs} can be generated by equality comparisons and ordinary array scans, without suffix trees, suffix arrays, LCE queries, or alphabet ordering.

Let $I=[b,e]$ be a recursion interval, let $m=\lfloor(b+e)/2\rfloor$, and put
\[
        U=w[b..m],\qquad V=w[m+1..e].
\]
Let $\bar U$ and $\bar V$ denote the reversals of $U$ and $V$, respectively.
For two words $A$ and $B$, define \textsc{PrefixScan}$(A,B)$ as the standard linear Z-scan on $A\#B$, where $\#$ is a fresh separator.
It fills an array $Z_{A,B}$ such that
\[
        Z_{A,B}[t]
        = \max\{\ell\ge 0 : A[1..\ell]=B[t..t+\ell-1]\},
\]
for all valid starting positions $t$ of $B$.
This scan compares characters only for equality and performs a left-to-right scan of an ordinary array.

For a fixed period $p$, recall the two anchors used in Lemma~\ref{lem:ml-crossing-representation}:
\[
        m-p+1,
        \qquad
        m.
\]
Let $u=|U|$ and $v=|V|$, and regard the positions of $U,V,\bar U,\bar V$ as $1$-based.
An out-of-range entry of a prefix-scan array is interpreted as $0$.

The anchor $m-p+1$ lies in the left half $U$ whenever $1\le p\le u$.
In local coordinates, it is the first position of the suffix $U[u-p+1..u]$.
Let
\[
        \alpha_p = Z_{V,U}[u-p+1],
        \qquad
        \beta_p  = Z_{V,V}[p+1],
        \qquad
        \gamma_p = Z_{\bar U,\bar U}[p+1].
\]
Here $\alpha_p$ compares the suffix $U[u-p+1..u]$ with the prefix of $V$.
If $\alpha_p=0$, then $m-p+1$ is not a $p$-good position.
If $\alpha_p>0$, the maximal $p$-good block containing $m-p+1$ extends $\gamma_p$ positions to the left.
Its right extension is $\alpha_p-1$ if $\alpha_p<p$, and is $p+\beta_p-1$ if $\alpha_p=p$.
Thus
\[
        \mathsf{Lext}_p(m-p+1)=\gamma_p,
        \qquad
        \mathsf{Rext}_p(m-p+1)=
        \begin{cases}
        \alpha_p-1, & \alpha_p<p,\\
        p+\beta_p-1, & \alpha_p=p.
        \end{cases}
\]
This determines the candidate induced by $m-p+1$ in constant time.

The anchor $m$ is relevant whenever $1\le p\le v$.
The first $p$ backward comparisons across the split compare the reversal of the suffix of $U$ with the reversal of the prefix $V[1..p]$.
This prefix of $\bar V$ starts at position $v-p+1$.
Let
\[
        \delta_p = Z_{\bar U,\bar V}[v-p+1],
        \qquad
        \eta_p   = Z_{\bar U,\bar U}[p+1],
        \qquad
        \theta_p = Z_{V,V}[p+1].
\]
If $\delta_p=0$, then $m$ is not a $p$-good position.
If $\delta_p>0$, its right extension inside $V$ is $\theta_p$.
Its left extension is $\delta_p-1$ if $\delta_p<p$, and is $p+\eta_p-1$ if $\delta_p=p$.
Thus
\[
        \mathsf{Rext}_p(m)=\theta_p,
        \qquad
        \mathsf{Lext}_p(m)=
        \begin{cases}
        \delta_p-1, & \delta_p<p,\\
        p+\eta_p-1, & \delta_p=p.
        \end{cases}
\]
Again, the candidate induced by $m$ is obtained in constant time.

Only a constant number of prefix-scan arrays over words of total length $O(|I|)$ are needed at node $I$.
After these arrays have been filled, the algorithm scans the periods $p=1,2,\ldots,\lfloor |I|/2\rfloor$.
For each $p$, it reconstructs the maximal $p$-good block containing $m-p+1$, if $m-p+1$ is a valid $p$-good position, and the maximal $p$-good block containing $m$, if $m$ is a valid $p$-good position.
Each such block $[a,c]$ yields the candidate $\langle a,c+p,p\rangle$ exactly when $c-a+1\ge p$ and the induced interval crosses the split.
If the two anchors produce the same block, one copy is discarded.
All operations after the prefix scans are ordinary integer and array operations, and hence the local candidate lists $\mathsf{Cand}_p(I)$ are generated in $O(|I|)$ time using equality tests only.

%


\bibliographystyle{cas-model2-names}
\bibliography{ref}

\end{document}